%% file: tensors-graphs.tex
\documentclass[submission,copyright,creativecommons]{eptcs}
\providecommand{\event}{QPL 2015}

\input{preamble.tex}

\title{Encoding !-tensors as !-graphs with neighbourhood orders}
\author{David Quick
  \institute{University of Oxford}
  \email{david.quick@cs.ox.ac.uk}
}
\def\titlerunning{!-tensors as !-graphs}
\def\authorrunning{D. Quick}

\begin{document}
\maketitle

\begin{abstract}
  Diagrammatic reasoning using string diagrams provides an intuitive language for reasoning about morphisms in a symmetric monoidal category. To allow working with infinite families of string diagrams, !-graphs were introduced as a method to mark repeated structure inside a diagram. This led to !-graphs being implemented in the diagrammatic proof assistant Quantomatic. Having a partially automated program for rewriting diagrams has proven very useful, but being based on !-graphs, only commutative theories are allowed. An enriched abstract tensor notation, called !-tensors, has been used to formalise the notion of !-boxes in non-commutative structures. This work-in-progress paper presents a method to encode !-tensors as !-graphs with some additional structure. This will allow us to leverage the existing code from Quantomatic and quickly provide various tools for non-commutative diagrammatic reasoning.
\end{abstract}

\tableofcontents

\newpage

\section{Background}\label{sec:background}

\subsection{Motivation}\label{sub:motivation}

Reasoning using diagrams can often be simpler and more intuitive than using a term-based mathematical syntax~\cite{JS}. There has recently been a shift toward using graphical calculi in a number of different fields. The \textit{String graph} formalism is one which has far-reaching applications in areas including categorical quantum mechanics~\cite{CD1,ContPhys,CDKZ}, computational linguistics~\cite{DimitriDPhil} and control theory~\cite{SobocinskiSignal,Baez2014a}. String graphs are a formalism for typed directed graphs, with the idea of edges replaced by (directed) wires. Unlike edges, wires do not need to have vertices at the ends. Those without starting vertices can be considered inputs and those without end vertices can be considered outputs. Wires also allow for identity edges (neither end connected to a node) or loops (a wire connected to itself). String graphs achieve this by having wire-vertices along a wire, so that a loop for example could be two wire-vertices $a$ and $b$ with edges $a\rightarrow b$ and $b\rightarrow a$. String graphs can be combined by plugging outputs from one in to inputs of another.

A diagrammatic theory is a set of equations between string diagrams on some set of generating morphisms. For example, a (commutative) associative multiplication operation %
\beginpgfgraphicnamed{Induc-Multiply}
\InputIfFileExists{Induc-Multiply.tikz}{}{\input{./figures/Induc-Multiply.tikz}}
\endpgfgraphicnamed with a unit %
\beginpgfgraphicnamed{Induc-Unit}
\InputIfFileExists{Induc-Unit.tikz}{}{\input{./figures/Induc-Unit.tikz}}
\endpgfgraphicnamed can be described by the string diagram equations:

\vspace{-5pt}
\ctikzfig{comm_string_eqns_eg}

An equation, $G=H$, between string graphs can be used in a form of rewriting where a subgraph matching $G$ can be replaced by $H$.

\vspace{-5pt}
\ctikzfig{rewriting_eg}

Many theories allow for nodes to have variable arity edges. For example, the multiplication and unit above can be replaced by a single operation taking the product of an arbitrary number of inputs (order unimportant), along with two diagram equations.

\vspace{-5pt}
\ctikzfig{arbit_eqns_eg}

This extension, while powerful, loses the mathematical rigour from the previous formalism (as can be seen from the use of ellipses). To return to a rigorous semantics, !-boxes (pronounced bang-boxes) were introduced in~\cite{DixonDuncan2009} to designate (by enclosing in a blue box) parts of a diagram which are allowed to be repeated. Adding !-boxes to string graphs results in the notion of a \textit{!-graph}. This was formalised in \cite{PatternGraph,MerryThesis} and has many applications, particularly in categorical quantum mechanics. Note that a !-graph represents a family of string graphs which can be recovered by two !-box operations, $\Kill_B$ and $\Exp_B$. $\Kill_B$ is the operation removing the !-box $B$ and all contents from a !-graph. $\Exp_B$ is the operation which creates a new copy of the contents of $B$ attached to the same nodes. Equations between !-graphs have the restriction that the !-graphs have the same !-boxes with the same nesting (!-boxes can contain other !-boxes). Hence the ellipsis equation in our example theory describes the following set:

\[
\left\llbracket %
\beginpgfgraphicnamed{box_eqn}
\InputIfFileExists{box_eqn.tikz}{}{\input{./figures/box_eqn.tikz}}
\endpgfgraphicnamed \quad \right\rrbracket
\ :=\  \left\{ \ \ %
\beginpgfgraphicnamed{box_eqn_inst}
\InputIfFileExists{box_eqn_inst.tikz}{}{\input{./figures/box_eqn_inst.tikz}}
\endpgfgraphicnamed \ \ \right\} 
\]

As !-graphs get applied to more complicated systems, rewriting by hand becomes difficult. This problem was solved by the introduction of Quantomatic~\cite{Quantomatic}, an automated theorem prover which allows rewriting of !-graphs using substitution. Quantomatic encompasses many tools for automated reasoning with string diagrams and has the ability to output graphical derivations directly to LaTeX. 

\vspace{-8pt}
\begin{figure}[H]\label{text}
  \centering
  \setlength{\abovecaptionskip}{-4pt}
  \includegraphics[scale=0.45]{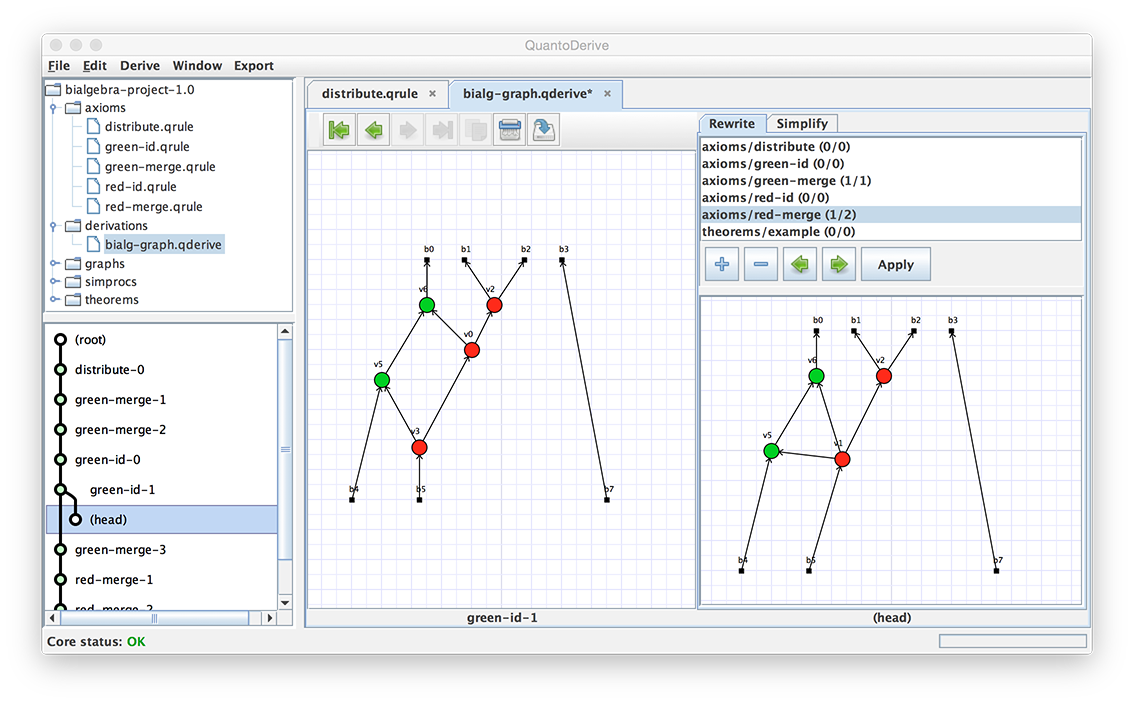}
  \caption{Example derivation in Quantomatic}
\end{figure}

One major drawback of both !-graphs and Quantomatic is that they are only designed to allow theories in which all nodes are commutative on inputs and outputs. This not only restricts the types of theory we can work with but more subtly, it rules out the option to definitionally extend a theory. i.e. defining new nodes (possibly inductively) as diagrams of other nodes. This is important if we want to rigorously replace finite arity nodes with their more intuitive arbitrary arity counterparts.

Since graphs inherently lack any notion of order on the edges surrounding a node, a more logical choice for representing non-commutative systems is a tensor notation in which edges are listed (in order) in a subscript. This tensor syntax was suggested in~\cite{KissingerATS} and then extended in~\cite{NoncommBB} to !-tensors which use an enriched version of Penrose's abstract tensor notation~\cite{Penrose1971} allowing non-commutative nodes and !-boxes. Here nodes are of the form $\phi_e$ where $\phi$ is the type and $e$ is called the edgeterm (describes order of edges). Nodes are enclosed in $[\_]^B$ to represent that they are inside a !-box named $B$. Edgeterms contain edges out $\+a$ and edges in $\<b$ where $a$ and $b$ are the names of the edges. If $\+a$ and $\<a$ appear in a !-tensor it represents an edge between their nodes. In the edgeterms anything inside $\rexp{\_}^B$ is an edge entering $B$ which should be expanded clockwise during $\Exp_B$ and $\lexp{\_}^B$ represents edges with anticlockwise expansion.

\vspace{-8pt}

\begin{equation}\label{eq:contraction}
  \psi_{\+f\bR\<a\<b\e}\phi_{\bR\+a\+b\e\<c\<d\<e}
  \quad := \quad
\beginpgfgraphicnamed{tensor_eg1}
\InputIfFileExists{tensor_eg1.tikz}{}{\input{./figures/tensor_eg1.tikz}}
\endpgfgraphicnamed
  \quad,\qquad \qquad
  [\phi_{\+a\<c\+b}\psi_{\+c\<d}]^B\xi_{\rexp{\<a}^B}\zeta_{\lexp{\<b\+d}^B\<e}
  \quad := \quad
\beginpgfgraphicnamed{tensor_eg2}
\InputIfFileExists{tensor_eg2.tikz}{}{\input{./figures/tensor_eg2.tikz}}
\endpgfgraphicnamed
\end{equation}

While it would be possible to write new automation tools based on term rewriting in the !-tensor formalism, this would require rebuilding the features already present in~\cite{Quantomatic} from scratch. Instead we wish to leverage this pre-existing code but extend it to work with the logic of non-commutative theories as formalised in~\cite{BangLogic}, by adding edgeterms to !-graphs. This paper shows work in progress demonstrating how we can fully encode !-tensors as !-graphs with added data. For the remainder of this section we recall the definitions of string graphs and !-graphs as defined in~\cite{MerryThesis}.

\subsection{String graphs as typed graphs}\label{sub:typed_string_graphs}

The standard definition of a graph in mathematics is not powerful enough for our purposes. As can be seen from the string graphs in Section~\ref{sub:motivation} our vertices need to be labelled (as generator types or wire-vertices). There should also be some restrictions as to which vertex types edges can connect. We achieve this using a typegraph, where every vertex gets mapped to a type.

\begin{definition}
  A graph $G$ along with a graph morphism $\tau:G\rightarrow \mathcal{G}$ is said to be a \textit{$\mathcal{G}$-typed graph}.
\end{definition}

\begin{definition}
  If $O$ and $M$ are sets and $\dom,\cod:M \rightarrow (O\times\{\bullet,\infty\})^*$ are functions into lists of pairs of elements of $O$ and either $\bullet$ or $\infty$, then $(O, M, \dom, \cod)$ is a \textit{compressed monoidal signature}.
\end{definition}

We refer to the set $O$ as the objects, the set $M$ as the morphisms and $\dom$ and $\cod$ assign the domain and codomain of each morphism allowing edges to be tagged as single ($\bullet$) or variable ($\infty$) arity.

\begin{definition}
  Given a compressed monoidal signature $T=(O,M,\dom,\cod)$, the \textit{derived compressed typegraph} $\mathcal{G}_T$ has vertices $O\cup M$ , a self-loop for every $X\in O$ and, for every $f\in M$,
  \begin{itemize}
    \item one $\bullet$-edge from $X$ to $f$ for each $(X, \bullet)$ in $\dom(f)$,
    \item one $\bullet$-edge from $f$ to $X$ for each $(X, \bullet)$ in $\cod(f)$,
    \item one $\infty$-edge from $X$ to $f$ for each $(X, \infty)$ in $\dom(f)$, and
    \item one $\infty$-edge from $f$ to $X$ for each $(X, \infty)$ in $\cod(f)$.
  \end{itemize}
\end{definition}

For example, the compressed monoidal signature $T$ with the two morphisms $f:[X^\infty]\rightarrow[X^\bullet]$ and $g:[Y^\bullet,X^\bullet]\rightarrow[X^\bullet]$ results in a typegraph $\mathcal{G}_{T}$ of the form:

\ctikzfig{typegraph_eg}

If a graph $G$ is $\mathcal{G}_T$-typed then we write $W(G)$ for the wire-vertices (those mapped into $O$) and $N(G)$ for the node-vertices (those mapped into $M$). $\bullet$-tagged edges represent edges of fixed (single) arity whereas $\infty$-tagged edges represent edges of variable arity. 

\begin{definition}
  A $\mathcal{G}_{T}$-typed graph $(G, \tau)$ is called a \textit{string graph} if $\tau$ is arity-matching (a bijection between the $\bullet$-edge neighbourhoods of $N$ and $\tau(N)$, for each $N\in N(G)$) and each wire-vertex in $G$ has at most one incoming edge and at most one outgoing edge.
\end{definition}

The notion of a string graph requires node vertices to be connected by chains of wire vertices which we refer to as wires. Wire points with no incoming edge represent an input to the graph and wire points with no outgoing edge represent outputs. We refer to the input vertices and output vertices together as the boundary. The set of wire-vertices in a wire are referred to as the interior of the wire.

\begin{remark}
  The edges around a node-vertex in a string graph are either $\bullet$-tagged, which means they are single arity; or they are $\infty$-tagged meaning there could be many copies made. This is important when we transition to !-graphs where we need to be careful not to copy single arity edges. 
\end{remark}

\subsection{!-Graphs}\label{sub:!graphs}

!-boxes denote subsections of a graph which can be repeated many times. We represent a !-box $B$ as a node (of the new type $!$) with an edge to each vertex contained in $B$.

\begin{definition}
  A subgraph $O$ of a string graph $G$ is said to be \textit{open} if it is not adjacent to any wire-vertices or incident to any fixed-arity edges.
\end{definition}

Openness (as described in \cite{MerryThesis}) encapsulates the property of being able to repeat that part of a graph an arbitrary number of times. It ensures adjacent edges are copied with nodes; fixed arity edges are not copied without the adjacent node; and that wires cannot be partially inside a !-box (to avoid wire splitting).

\begin{definition}
  Given a compressed monoidal signature $T$, the \textit{derived compressed !-typegraph} $\mathcal{G}_{T!}$ is $\mathcal{G}_{T}$ with the addition of a vertex $!$ along with edges from $!$ to every vertex (including itself).
\end{definition}

Our previous example of the compressed monoidal signature $T$ with morphisms $f:[X^\infty]\rightarrow[X^\bullet]$ and $g:[Y^\bullet,X^\bullet]\rightarrow[X^\bullet]$ results in a !-typegraph $\mathcal{G}_{T!}$ of the form:

\begin{equation}
\beginpgfgraphicnamed{btypegraph_eg}
\InputIfFileExists{btypegraph_eg.tikz}{}{\input{./figures/btypegraph_eg.tikz}}
\endpgfgraphicnamed
\end{equation}

If G is $\mathcal{G}_{T!}$-typed, write $!(G)$ for the box-vertices in $G$ (those mapped to $!$) and $U(G)$ for the full subgraph of $G$ with all vertices except $!(G)$. Given $B\in !(G)$ let $C(B)$ be the full subgraph of $G$ with nodes which have edges from $B$. So $C(B)$ represents the contents of a !-box $B$. We use this notation in the !-graph conditions to ensure !-boxes are well-behaved under !-box operations.

\begin{definition}\label{def:!-graph}
  A $\mathcal{G}_{T!}$-typed graph $G$ is called a \textit{!-graph} if the following hold:
  \begin{itemize}
  \setlength{\itemindent}{.25in}
    \item[BG1.] $U(G)$ is a $\mathcal{G}_{T}$-typed string graph;
    \item[BG2.] the full subgraph with vertices $!(G)$ is posetal;
    \item[BG3.] $\forall B\in !(G)$, $U(C(B))$ is an open subgraph of $U(G)$; and
    \item[BG4.] $\forall B,B'\in !(G)$, if $B'\in C(B)$ then $C(B')\subseteq C(B)$.
  \end{itemize}
\end{definition}

As an example, the theory described in Section~\ref{sub:motivation} has !-typegraph and !-graph rules as follows:

\ctikzfig{comm_string_eqns_eg_typed}

\section{!-Graphs with simple overlap and ordered neighbourhoods}\label{sec:graphs_so_no}

Given a compact closed signature $\Sigma$ we wish to encode the !-tensors of $\sgraphterms$ as !-graphs with some additional data. However, !-graphs allow for two or more non-nested !-boxes to share nodes, which is not permitted with !-tensors. In this section we define a restricted set of !-graphs and the additional structure recording non-commutativity.

\subsection{Simple overlap}\label{sub:simple_overlap}

We must now add the restriction to rule out !-graphs with overlap on nodes, we do this by introducing the idea of simple overlap of !-boxes.

\begin{definition}
  Given a pair of non-nested !-boxes $B$ and $B'$, we say they \textit{overlap simply} if $C(B)\cap C(B')$ consists of only the interior of zero or more wires, where at least one endpoint is a node-vertex and any node-vertex endpoints are in either $C(B)$ or $C(B')$.
\end{definition}

\begin{definition}
  A !-graph where any two non-nested !-boxes overlap simply is called a \textit{!-graph with simple overlap}.
\end{definition}

Compare this definition to !-graphs with trivial overlap (BGTO) which were first defined in \cite{bgto} where it was shown that they can be encoded using a context-free grammar. The difference here is that we allow non-nested overlap at a boundary rather than just between two node-vertices.

\ctikzfig{simple_overlap_eg}

The first two !-graphs have only simple overlap (though the second is non-trivial) the final two have non-simple overlap.

\subsection{Adding orders to neighbourhoods}\label{sub:adding_orders_to_neighbourhoods}

We now only need to describe the additional data to keep track of edge orders. For this we need the notion of an edgeterm which was originally described in \cite{NoncommBB}. Edgeterms list, in clockwise order, the edges (in $\ein a$ or out $\eout a$) around a node grouping those which expand clockwise under !-box $B$ inside $\rexp{\_}^B$ and those which expand anticlockwise inside $\lexp{\_}^B$. For examples see \eqref{eq:contraction} and \cite{NoncommBB}.

\begin{definition}\label{def:edgeterm}
The set of \textit{edgeterms} $\edgetermson{\mathcal{E}}$ over the edge names $\mathcal{E}$ and !-box names $\boxnames$ is defined inductively as follows:
  \begin{align*}
    \bullet\; &\epsilon \in \edgetermson{\mathcal{E}} && \text{(empty edgeterm)} \\
    \bullet\; & \ein a, \eout a \in \edgetermson{\mathcal{E}} && a \in \mathcal{E} \\
    \bullet\; &\lexp{e}^A,\rexp{e}^A \in \edgetermson{\mathcal{E}} && e\in\edgetermson{\mathcal{E}},\; A\in\boxnames \\
    \bullet\; & e f \in \edgetermson{\mathcal{E}} && e,f\in\edgetermson{\mathcal{E}}
  \end{align*}
    Two edgeterms are equivalent if one can be transformed into the other by:
  \begin{equation}
    e(fg) \equiv (ef)g \qquad
    \epsilon e \equiv e \equiv e \epsilon \qquad
    \rexp{\epsilon}^A \equiv \epsilon \equiv \lexp{\epsilon}^A
  \end{equation}
\end{definition}

\begin{definition}\label{def:nhd_order}
  Given a !-graph with simple overlap $G$, a \textit{neighbourhood order} on $G$ is a function $\nhd: N(G)\rightarrow \edgetermson{W(G)}$ satisfying $\forall N\in N(G)$:
  \begin{itemize}
    \item $\nhd(N)$ is an edgeterm with input edge names $\{w\in W(G): w\rightarrow N\}$ and output edge names $\{w\in W(G): N\rightarrow w\}$,
    \item The !-boxes with edges to $a\in\nhd(N)$ but not to $N$ are precisely those containing $a$ in $\nhd(N)$ (ordered by nesting order).
  \end{itemize}
\end{definition}

We define three !-box operations: $\Kill_B$ removes $B$ and its contents; $\Drop_B$ removes $B$ leaving its contents; and $\Copy_B$ add a new copy of $B$ and its contents. The operation $\Exp_B$ can be applied using $\Copy_B$ and then dropping the new !-box. 

\begin{definition}\label{def:operations}
  For $\Op_B$ any !-box operation, $\Op_B(G,\nhd):=(\Op_B(G),\Op_B\circ\nhd)$. Where our three !-box operations are applied to !-graphs by:
  \begin{itemize}
    \item $\Copy_B(G)$ is defined by a pushout of inclusions:
    \ctikzfig{copy_pushout}
    \item $\Drop_B(G):=G\backslash B$
    \item $\Kill_B(G):= G\backslash C(B)$
  \end{itemize}
  and $\Op_B$ is applied to an edgeterm in the following trivial cases:
  \begin{align*}
    \Op_B(\rexp{e}^A) &:= \rexp{\Op_B(e)}^A &
    \Op_B(e f) &:= \Op_B(e) \Op_B(f) \\
    \Op_B(\lexp{e}^A) &:= \lexp{\Op_B(e)}^A &
    \Op_B(x) &:= x
  \end{align*}
  where $A \neq B$ and $x \in \{ \ein a, \eout a, \epsilon \}$ and for the final two cases:
  \begin{align*}
    \Copy_B(\rexp{e}^B)  &:= \rexp{e}^B \rexp{\fr(e)}^{\fr(B)} &
    \Drop_B(\rexp{e}^B)  &:= e &
    \Kill_B(\rexp{e}^B) &:= \epsilon \\
    \Copy_B(\lexp{e}^B)  &:= \lexp{\fr(e)}^{\fr(B)} \lexp{e}^B &
    \Drop_B(\lexp{e}^B)  &:= e &
    \Kill_B(\lexp{e}^B) &:= \epsilon \\
  \end{align*}
\end{definition}

\section{Encoding !-tensors}\label{sec:encoding}

\subsection{Defining the mapping \texorpdfstring{$\I$}{I}}\label{sub:defining_the_mapping}

We wish to allow theories with generating morphisms which have input or output edges of single ($\bullet$-tagged) or variable ($\infty$-tagged) arity. Given a set $\mathcal O$ of object types and $X\in\mathcal O$ we write $\<X^\bullet,\+X^\bullet,\<X^\infty,\+X^\infty$ to represent fixed arity input and output edges and variable arity input and output edges respectively. We need a different definition of a signature when working with !-tensors.

\begin{definition}\label{def:!tensor_signature}
  A \textit{compact closed signature} $\Sigma$ consists of a set $\mathcal O := \{ X, Y, \ldots \}$ of object types and a set $\mathcal M$ of pairs $(\psi, w)$, where $w$ is a word in $\{\<X^\bullet,\+X^\bullet,\<X^\infty,\+X^\infty:X\in\mathcal O\}$.
\end{definition}

Recall the definition of a !-tensor expression and equivalence of !-tensors from~\cite{NoncommBB}.

\begin{definition}\label{def:!-tensors}
The set of all !-tensor expressions $\sgraphterms$ for a signature $\Sigma$ is defined recursively as:
  \begin{align*}
    \bullet\; &1, 1_{\+a\<b} \in \sgraphterms && a,b\in\edgenames\\
    \bullet\; &\phi_{e} \in \sgraphterms && e\in\edgetermson{\mathcal{E}}, \phi\in\Sigma \\
    \bullet\; &[G]^A \in \sgraphterms && G\in\sgraphterms, \; A\in\boxnames \\
    \bullet\; & G H \in \sgraphterms && G,H\in\sgraphterms
  \end{align*}
  Where $\mathcal E$ is a set of possible edge names and $\mathcal B$ is a set of possible !-box names. Subject to the conditions (F1) $\<a$ and $\+a$ must occur at most once for each edge name $a$ and (F2) $[\ldots]^A$ must occur at most once for each !-box name $A$, as well as some consistency conditions (C1)-(C3) for !-boxes~\cite{NoncommBB}.
\end{definition}

\begin{definition}\label{def:tensor-equiv}
  Two !-tensor expressions $G$ and $H$ are considered equivalent, $G \equiv H$, if one can be obtained from the other by replacing bound names and/or applying one or more of the following enforced identities:
  \[
    (GH)K \equiv G(HK) \qquad
    GH \equiv HG \qquad
    G1 \equiv G
  \]
  \[
    \Weak_{B\ni 1_{\+b\<a}}(G) \equiv G[\<b \mapsto \<a] \qquad
    \Weak_{B\ni 1_{\+a\<b}}(H) \equiv H[\+b \mapsto \+a]
  \]
  where for the last two identities $\< b$ and $\+ b$ are free in $G$ and $H$ respectively. Note $\Weak$ is inserting the identity wires inside the appropriate !-box $B$ such that the edge $b$ is well defined.
\end{definition}

\begin{definition}
  Given a compact closed signature $\Sigma=(\mathcal{O},\mathcal{M})$, we define the corresponding compressed monoidal signature $\I(\Sigma)$ to have objects $\mathcal{O}$, morphisms $\{\phi : (\phi,w)\in\mathcal{M}\}$, $\dom$ is defined to map $\phi$ to the input edges of $w$ for each $(\phi,w)\in\mathcal{M}$ and similarly $\cod$ maps $\phi$ to the output edges of $w$.
\end{definition}

For example, the !-tensor signature $\Sigma$ with objects $\{X,Y\}$ and morphisms $\{(f,\+X^\bullet\<X^\infty),(g,\+X^\bullet\<Y^\bullet\<X^\bullet)\}$ becomes the !-graph signature with the two morphisms $f:[X^\infty]\rightarrow[X^\bullet]$ and $g:[Y^\bullet,X^\bullet]\rightarrow[X^\bullet]$. This results in a !-typegraph $\mathcal{G}_{\I(\Sigma)!}$ of the form:

\ctikzfig{btypegraph_eg}

For the rest of this section we will suppose we are working in a fixed signature $\Sigma=(\mathcal{O},\mathcal{M})$.

The !-tensor formalism avoids the need to name nodes, but we will need to assign them with unique names if we want to convert to !-graphs. We do this using an indexing set. Say $G$ is a !-tensor and let $J$ be a set indexing the nodes (subexpressions of the form $\phi_e$, $1_{\+b\<a}$ or $1$) of $G$, write $N_j$ for the node corresponding to $j\in J$. In converting to a !-graph, the node $N_j$ can be labelled $j$.

We can define a function $\I$ taking !-tensors to what we will prove to be !-graphs with simple overlap and neighbourhood orders:

\begin{definition}
  $\I$ taking !-tensor expressions to !-graphs with neighbourhood orders is defined recursively:
  \begin{itemize}
    \item $\I(1):=\{\}$
    \item $\I(1_{\+b\<a}):=\{a\rightarrow b\}$
    \item $\I(\phi_e=N_j):=\{a\rightarrow j : \<a\in e\}\cup\{j\rightarrow a : \+a\in e\}\cup\{\bB B\rightarrow \e a : B\in\ectx_e(a)\}$
    \item $\I(GH):=\I(G)\cup\I(H)$
    \item $\I([G]^B):=\I(G)\cup\{\bB B\rightarrow \e x : x\in U(\I(G))\}\cup\{\bB B\rightarrow B'\e : B'\leq B\}$
  \end{itemize}
  The typing function on vertices maps $j\rightarrow \phi$ where $N_j=\phi_e$, !-boxes to $!$ and edge names to their edge type. The neighbourhood order is $nhd(j):=e$ for each $N_j=\phi_e$.
\end{definition}

So, for example, from the tensor expression $\phi_{\lexp{\<a}^A}[\psi_{\+a\rexp{\<b}^B}[1]^B]^A1_{\+d\<c}$ if we index the nodes from left to right using $\{1,2,3,4\}$ then we get the following after $\I$:

\vspace{-5pt}

\begin{align*}
  \I(N_1)&=%
\beginpgfgraphicnamed{string_eg_1}
\InputIfFileExists{string_eg_1.tikz}{}{\input{./figures/string_eg_1.tikz}}
\endpgfgraphicnamed &
  \I(N_2)&=%
\beginpgfgraphicnamed{string_eg_2}
\InputIfFileExists{string_eg_2.tikz}{}{\input{./figures/string_eg_2.tikz}}
\endpgfgraphicnamed &
  \I(N_3)&=%
\beginpgfgraphicnamed{string_eg_3}
\InputIfFileExists{string_eg_3.tikz}{}{\input{./figures/string_eg_3.tikz}}
\endpgfgraphicnamed &
  \I(N_4)&=%
\beginpgfgraphicnamed{string_eg_4}
\InputIfFileExists{string_eg_4.tikz}{}{\input{./figures/string_eg_4.tikz}}
\endpgfgraphicnamed
\end{align*}
Going through the recursive definition these combine to the !-graph we would expect:

\ctikzfig{string_eg_from_tensor}

\subsection{\texorpdfstring{$\I$}{I} is Well-Behaved}\label{sub:i_well_behaved}

We now check $\I$ behaves as desired so we can use this mapping to encode !-tensors in Quantomatic.

\begin{theorem}\label{thm:i_bgsono}
  Given $G\in\sgraphterms$, $\I(G)$ is a $\mathcal{G}_{\I(\Sigma)!}$-typed !-graph with simple overlap and a neighbourhood order.
\end{theorem}
\begin{proof}
    Let $J$ index the nodes in $G$.
  \begin{itemize}
    \item We first go through the four conditions of Definition~\ref{def:!-graph} to check we have a $\mathcal{G}_{\I(\Sigma)!}$-typed !-graph:
    \begin{enumerate}
      \item[BG1.]
      Take $(\phi,w)$ a morphism in $\I(\Sigma)$. The typegraph $\mathcal{G}_{\I(\Sigma)}$ contains a node $\phi$ with $\bullet$-edges to/from the $\bullet$-tagged edges in $w$. Given $\phi_e=N_j$, the $\bullet$-edge neighbourhood of $j$ is the finite arity edges in $e$ which are the same as $\bullet$-tagged edges in $w$. Hence we have a bijection. \\
      From the definition of $\I$, wire vertices come directly from edge names. The unique (possible) occurrence of $\+a$ results in a wire $\_\rightarrow a$ and the unique (possible) occurrence of $\<a$ results in a wire $a\rightarrow \_$. Hence each $a$ must have at most one incoming and one outgoing edge.
      \item[BG2.]
      The full subgraph $!(G)$ is the reflexive, transitive closure of the `child of' relation on !-boxes. Hence it is reflexive, transitive and antisymmetric (ensuring `child of' is cycle-free).
      \item[BG3.]
      Take $B\in !(G)$ and write $X$ for $U(C(B))$, so we need to show that $X$ is an open subgraph of $U(G)$. We first show that any wire-vertex adjacent to a vertex in $X$ is in $X$. From the definition of $\I$, the only time an edge is added from a box-vertex to a node-vertex $j$, there are also edges added to each neighbour of $j$. Any two adjacent wire-vertices must come from $\I(1_{\+a\<b})$, so that $\bB B\rightarrow \e a\Leftrightarrow \bB B\rightarrow \e b$. \\
      Incident edges can only come from node-vertices not in $X$ with adjacent wire-vertices in $X$. Hence, the wire comes from a directed edge with $B$ in its edge context which means it is $\infty$-tagged in the typegraph.
      \item[BG4.]
      An edge between box-vertices, $\bB A\rightarrow B\e$ must be added by $\I([H]^A)$ during the recursive definition, where $B$ and $C(B)$ are already in $\I(H)$. Hence we get edges $\bB A \rightarrow \e x$ for all $x\in C(B)$ meaning $C(B)\subseteq C(B')$.
    \end{enumerate}
    \item Next we check that any overlap of non-nested !-boxes $B,B'$ is simple.
    \subitem Suppose $j$ is a node-vertex such that $\bB B\rightarrow \e j$ and $\bB B'\rightarrow \e j$. From the definition of $\I$ we see that $B$ and $B'$ must appear in the node context of $N_j$ and so one is nested inside the other. This contradiction proves that $C(B)\cap C(B')$ contains only wire vertices
    \subitem Suppose $a$ is a wire-vertex such that $\bB B\rightarrow \e a$ and $\bB B'\rightarrow \e a$. If $a$ is adjacent to another wire-vertex $b$ then this edge must have come from $1_{\+b\<a}$ or $1_{\+a\<b}$ so $a$ and $b$ are nested inside the same !-boxes. Hence the intersection $C(B)\cap C(B')$ contains only the interiors of zero or more wires.
    \subitem Suppose both endpoints are wire-vertices. Hence the edges have come from identity wires $1_{\+b\<a}$ which must have $B$ and $B'$ in their node context which would require them to be nested.
    \subitem Suppose a wire-vertex $a\in C(B)\cap C(B')$ is adjacent to the node-vertex $j$. If $j$ was not in $C(B)$ nor $C(B')$ then both $B$ and $B'$ must occur in the edge context of $a$ and hence would be nested. We conclude that $j$ must occur in exactly one of $C(B)$ or $C(B')$.
    \item Finally we wish to check that $\nhd$ is a neighbourhood order on $\I(G)$. For $j\in N(\I(G))$ the incoming edges $a$ to $j$ come from inputs $\<a$ and the outgoing edges $b$ come from outputs $\+b$ as required. Also, for the node $N_j=\phi_e$, !-boxes with edges to $a\in\ectx{e}$ but not to $j$ are exactly those which appear in $\ectx_e(a)=\ectx_{\nhd(\phi)}(a)$.
  \end{itemize}
\end{proof}

Hence $\I$ takes a !-tensor expression and returns a correctly typed !-graph with a neighbourhood order. By the following theorem, the definition of $\I$ can be lifted from specific !-tensor expressions to !-tensors (equivalence classes of !-tensor expressions).

\begin{theorem}
  $\forall G,H\in\sgraphterms$, $G\equiv H$ if and only if $\I(G)$ and $\I(H)$ are equivalent up to renaming and wire homeomorphisms.
\end{theorem}
\begin{proof}
  We prove first that all !-tensor equivalences from Definitions~\ref{def:edgeterm} and \ref{def:tensor-equiv} are preserved through $\I$. In both formalisms bound variables can be renamed freely, so we will not worry about names here.
  \begin{itemize}
    \item Since edgeterms are copied directly from !-tensors to !-graphs, the edgeterm equivalences are still preserved. The associativity, commutativity and identity equivalences on !-tensor expressions have no affect on the graphical formalism so these are also preserved. The only thing left to check are the two equivalences involving inserting identity wires $1_{\+b\<a}$. These two conditions come down to wire homeomorphism in the !-graph framework. First we are given that $\<b$ exists in $G$ and (for some !-box $B$) we look at $\Weak_{B\ni 1_{\+b\<a}}(G)$ under $\I$. This becomes a graph with edges $a\rightarrow b \rightarrow x$ for some $x$ and since $a$ and $b$ are both wire-vertices this is wire homeomorphic to $G[\<b \mapsto \<a]$. The other case is similar but with arrows reversed.
    \item For the other direction suppose $\I(G)$ and $\I(H)$ differ only by a single wire homeomorphism, so there exist wire-vertices $a,b$ and a vertex $c$ in $\I(G)$ with $a\rightarrow b\rightarrow c$, and $\I(H)$ is $\I(G)$ but with $a\rightarrow b\rightarrow c$ replaced with $a\rightarrow c$. From the definition of $\I$ we see that there exists some $G'$ and !-box $B$ such that $G=\Weak_{B\ni 1_{\+b\<a}}(G')$ (this is the only way two wire vertices can be connected) and also $\<b$ must exist in $G$. $H$ must be the same as $G'$ except the edge $\<b$ is replaced by the edge $\<a$.
  \begin{align*}
    G &\equiv \Weak_{B\ni 1_{\+b\<a}}(G') \\
      &\equiv G'[\<b \mapsto \<a] \\
      &\equiv H
  \end{align*} 
  \end{itemize}
\end{proof}

We have shown $\I$ is injective on !-tensors and hence is a bijection onto its image. We can hence take any !-tensor $G$ and work with it in Quantomatic in the form of $\I(G)$. To work as an encoding we would hope that !-box operations are equivalent in each formalism. Applying the !-box operation $\Op_B$ to $\I(G)$ and then returning to the !-tensor formalism should result in a !-tensor equivalent to $\Op_B(G)$. By the previous theorem we need only check that $\Op_B\circ\I=\I\circ\Op_B$ and it then follows that $\I^{-1}\circ\Op_B\circ\I=\Op_B$.

\begin{theorem}
  $\Op_B(\I(G))=\I(\Op_B(G))$ for any !-box operation $\Op_B$ and $G\in\sgraphterms$.
\end{theorem}
\begin{proof}
  This can be shown by case analysis on the recursive definition of $\sgraphterms$ going through each operation $\Copy_B$, $\Drop_B$, $\Kill_B$. Most cases are trivial, the interesting case is showing $\Op_B(\I([G]^B))=\I(\Op_B([G]^B))$ for each operation:\\
    $\Copy_B(\I([G]^B))$ is defined by a pushout of the inclusion $1\hookrightarrow \I([G]^B)$ with itself, so equals the disjoint union of two copies of $\I([G]^B)$. Hence (for $\rn$ a renaming function):
    \begin{align*}
       \Copy_B(\I([G]^B))&=\I([G]^B)\cup\I([\rn(G)]^{\rn(B)}) \\
       &=\I([G]^B[\rn(G)]^{\rn(B)}) \\
       &=\I(\Copy_B([G]^B)) \\
      \Drop_B(\I([G]^B))&=\Drop_B(\I(G)\cup\{\bB B\rightarrow \e x : x\in U(\I(G))\}\cup\{\bB B\rightarrow B'\e : B'\leq B\}) \\
      &=\I(G) \\
      &=\I(\Drop_B([G]^B)) \\
      \Kill_B(\I([G]^B))&=\Kill_B(\I(G)\cup\{\bB B\rightarrow \e x : x\in U(\I(G))\}\cup\{\bB B\rightarrow B'\e : B'\leq B\}) \\
      &=\{\} \\
      &=\I(1) \\
      &=\I(\Kill_B([G]^B))
     \end{align*}
\end{proof}

\begin{corollary}
  The concrete instances of $\I(G)$ are the concrete instances of $G$ with $\I$ applied to them.
\end{corollary}

\section{Future work}\label{sec:future_work}

\subsection{Theory}\label{sub:theory}

This paper demonstrates work in progress and there are a few things which need to be added to the theory. Firstly, we would like to show that $\I$ is a bijection on to the set of !-graphs with simple overlap and neighbourhood order. This would allow us to replace Theorem~\ref{thm:i_bgsono} with an if and only if statement:

\begin{conjecture}
  $\exists G\in\sgraphterms$ such that $\I(G)=(X,\nhd)$ if and only if all of the following hold:
  \begin{itemize}
    \item $X$ is a $\mathcal{G}_{\I(\Sigma)!}$-typed !-graph
    \item $X$ has simple overlap
    \item $\nhd$ is a neighbourhood order on $X$
  \end{itemize}
\end{conjecture}

Related to this, we would like to explore the differences between simple overlap and the stricter property of trivial overlap of !-boxes~\cite{bgto}. 

Matching (finding occurrences of one !-graph inside another) and substitution (replacing matches with equivalent graphs) are the key parts of diagrammatic reasoning with !-graphs. Matching a string graph $G$ with a string graph $K$ means finding a monomorphism from $G$ in to $K$ showing that $G$ is isomorphic to a subgraph of $K$. For example, in a commutative theory the following is a string graph matching:
\ctikzfig{string_graph_matching}
This may not be the case if we are working with non-commutative nodes. If the top node-vertex in $G$ has edges ordered $ab$ and the top node-vertex in $K$ has edges ordered $yx$ then we can not find a matching. This can be seen from their !-tensor equivalents:
\ctikzfig{string_graph_matching_ceg}
Hence we hope to define a matching from $(G,\nhd_G)$ to $(K,\nhd_K)$ to be a !-graph matching $i:G\rightarrow K$ satisfying $i(\nhd_G(x))=\nhd_K(i(x))\ \forall x\in N(G)$. This is currently being looked into by the author, along with checking that substitution works without any changes.

Once matching and substitution are formalised we can start comparing the differences between rewriting using substitution in !-graphs with neighbourhood orders and !-logic as defined in~\cite{BangLogic}. Ideally we would like to demonstrate that any case of proof by substitution of !-graphs (with neighbourhood order) can be shown in !-logic for !-tensors and conversely that each rule in !-logic can be implemented using substitution.

\subsection{Implementation}\label{sub:implementation}

The key contribution of this paper is to demonstrate that !-tensors could be quickly implemented in a program such as Quantomatic by adding neighbourhood order data to node vertices. Hence an obvious piece of further work would be to actually implement this in Quantomatic. There are three key areas to work on here. First is to code the new neighbourhood data on the !-graph datatype. Secondly we should update the matching and substitution algorithms to respect neighbourhood orders. Finally the user needs to be able to interact with !-graphs in a new way allowing neighbourhood orders to be chosen. 

We add neighbourhood orders by recursively defining an edgeterm datatype (using the set of edge and !-box names) and then adding a node-vertex indexed table of edgeterms to the !-graph datatype. Code can be written to enforce the restrictions of Definition~\ref{def:nhd_order} whenever a !-graph is updated. The functions acting on !-graphs can then be extended to also work on the individual edgeterms. For example !-box operations should be defined to act on edgeterms as described in Definition~\ref{def:operations}. 

Once the theory has been checked, the matching and substitution algorithms can be updated to test orders on neighbourhoods are preserved. A naive approach to matching would be to look for all matches ignoring edgeterms and then to filter out those in which edgeterms do not match up. This would be easy to implement but not very efficient. More work is needed to find a better method, checking edgeterms as the matching algorithm searches. As for substitution, the requirement that two !-tensors in a !-tensor equation have compatible boundaries should make sure that substitution does not create any issues. 

We present a planned graphical user interface for !-tensors in Quantomatic in Appendix~\ref{sec:graphical_user_interface}.

Since non-commutativity is added to each node-vertex individually it may be possible for Quantomatic to have commutative and non-commutative nodes interacting in one diagram. This is something which needs to be explored further.

\newpage

\bibliographystyle{eptcs}
\bibliography{bibfile}

\newpage

\appendix

\section{Graphical user interface}\label{sec:graphical_user_interface}

Implementing the gui could require a large amount of work. The current method for drawing edges simply draws a line directly between nodes. This will need to be changed for non-commutative nodes where it becomes important where edges enter/leave a node. New edges will need to be curves defined by their source and target nodes and angles. We will add handles to each edge controlling which direction the edge should go. Figure~\ref{fig:edge_order} demonstrates what this may look like. We envisage Quantomatic choosing the angles initially and only letting the user drag one edge handle over another edge to reorder them. If the user wished to have more control, double clicking the handle would allow them to control the angle exactly. The user should also be able to choose the angle of the tick.

\begin{figure}[H]
\centering
\begin{minipage}{.5\textwidth}
  \centering
\beginpgfgraphicnamed{ncq_edge_order}
\InputIfFileExists{ncq_edge_order.tikz}{}{\input{./figures/ncq_edge_order.tikz}}
\endpgfgraphicnamed
  \caption{Reordering edges}
  \label{fig:edge_order}
\end{minipage}%
\begin{minipage}{.5\textwidth}
  \centering
\beginpgfgraphicnamed{ncq_arc_drawing_2}
\InputIfFileExists{ncq_arc_drawing_2.tikz}{}{\input{./figures/ncq_arc_drawing_2.tikz}}
\endpgfgraphicnamed
  \caption{Manipulating arcs}
  \label{fig:arc_drawing}
\end{minipage}
\end{figure}

Figure~\ref{fig:arc_drawing} demonstrates the second major interface modification. We need to draw and let the user interact with arcs over edges entering a !-box. Whenever the user draws an edge between a non-commutative node and a !-box, Quantomatic will automatically add an arc over the edge (grouping multiple edges if possible). To change the arcs the user should be able to click on any part of an arc and drag. This would start a new arc containing every edge the user drags over in the direction of dragging. Hence arcs over multiple edges can be split apart by drawing each single arc in its required direction.

\begin{figure}[h]
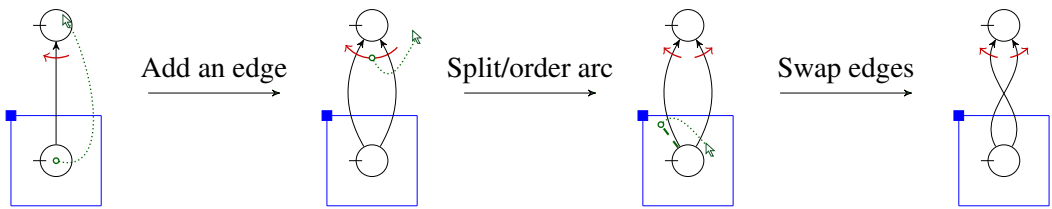
\label{fig:quanto_workflow}
  \ctikzfig{ncq_eg}
  \caption{An example workflow in Quantomatic}
\end{figure}

\end{document}

%% file: preamble.tex
\usepackage{tikzfig}
\input{tikzstyles.tex}

\usepackage{noncommgraph}

\usepackage{bussproofs}
\usepackage{amsmath}
\usepackage{amssymb}
\usepackage{stmaryrd}
\usepackage{amsthm}
\usepackage{xspace}
\usepackage{hyperref}
\usepackage{bussproofs}
\usepackage{forest}
\usepackage{graphicx}
\graphicspath{{figures/}}

\usepackage{algpseudocode}
\usepackage{algorithm}
\algblock[Switch]{Switch}{EndSwitch}
\algblock[Case]{Case}{EndCase}
\algtext*{EndCase}

\newtheorem{theorem}{Theorem}[section]
\newtheorem*{theorem*}{Theorem}

\newtheorem{corollary}[theorem]{Corollary}
\theoremstyle{definition}
\theoremstyle{definition}
\theoremstyle{definition}\newtheorem{definition}[theorem]{Definition}
\theoremstyle{definition}
\theoremstyle{definition}\newtheorem{remark}[theorem]{Remark}
\theoremstyle{definition}\newtheorem{conjecture}[theorem]{Conjecture}
\theoremstyle{definition}
\theoremstyle{definition}
\theoremstyle{definition}
\theoremstyle{definition}

\input{defs.tex}

\def\bR{\begin{color}{red}}
\def\bB{\begin{color}{blue}}
\def\bM{\begin{color}{magenta}}
\def\bC{\begin{color}{cyan}}
\def\bW{\begin{color}{white}}
\def\bBl{\begin{color}{black}}
\def\bG{\begin{color}{green}}
\def\bY{\begin{color}{yellow}}
\def\bDG{\begin{color}{green!70!black}}
\def\e{\end{color}}

\usepackage{wrapfig}
  {\gdef\scalefactor{1.0}\begin{center}\proofSkipAmount \leavevmode}%
  {\scalebox{\scalefactor}{\DisplayProof}\proofSkipAmount \end{center} }

%% file: tikzstyles.tex
\tikzstyle{every picture}=[baseline=-0.25em,scale=0.6]
\tikzstyle{small box}=[rectangle,draw,fill=white,minimum width=3mm,minimum height=3mm]
\tikzstyle{empty diagram}=[draw=gray!40!white,dashed,shape=rectangle,minimum width=8mm,minimum height=8mm]

\tikzstyle{dotpic}=[scale=0.6]
\tikzstyle{directeds}=[every to/.style={directed}]
\tikzstyle{dot graph}=[shorten <=-0.1mm,shorten >=-0.1mm,scale=0.6]
\tikzstyle{digraph}=[-latex]
\tikzstyle{plot point}=[circle,fill=black,minimum width=2mm,inner sep=0]
\tikzstyle{string graph}=[scale=0.6]
\tikzstyle{sg directed}=[-stealth]
\tikzstyle{rewrite edge}=[-open triangle 45]
\tikzstyle{sg bold directed}=[-stealth,thick,shorten >=-1pt]
\tikzstyle{sg vertex}=[circle,minimum width=2.2mm,fill=white,draw=black,inner sep=0mm]
\tikzstyle{labelled sg vertex}=[circle,minimum width=7mm,fill=white,draw=black,inner sep=0mm]
\tikzstyle{sg grey vertex}=[sg vertex,fill=gray!30!white]
\tikzstyle{sg black vertex}=[sg vertex,fill=black]
\tikzstyle{sg bold vertex}=[circle,minimum width=2.2mm,fill=white,draw=black,very thick,inner sep=0mm]
\tikzstyle{sg wire vertex}=[circle,minimum width=1mm,fill=black,inner sep=0mm]

\tikzstyle{tick vertex}=[rectangle,fill=black,minimum height=1mm,minimum width=2.5mm,inner sep=0mm]

\tikzstyle{braceedge}=[decorate,decoration={brace,amplitude=2mm,raise=-1mm}]
\tikzstyle{small braceedge}=[decorate,decoration={brace,amplitude=1mm,raise=-1mm}]
\tikzstyle{left hook arrow}=[left hook-latex]
\tikzstyle{right hook arrow}=[right hook-latex]

\tikzstyle{dot}=[inner sep=0.7mm,minimum width=0pt,minimum height=0pt,fill=black,draw=black,shape=circle]
\tikzstyle{white dot}=[dot,fill=white]
\tikzstyle{alt white dot}=[white dot,label={[xshift=2.9mm,yshift=-0.1mm]left:$\cdot$}]
\tikzstyle{gray dot}=[dot,fill=gray!50]
\tikzstyle{box vertex}=[draw=black,rectangle]
\tikzstyle{whitebg}=[fill=white,inner sep=2pt]
\tikzstyle{graph state vertex}=[sg vertex,fill=black]

\tikzstyle{wide point}=[fill=white,draw=black,shape=isosceles triangle,shape border rotate=90,isosceles triangle stretches=true,inner sep=1pt,minimum width=1.5cm,minimum height=5mm]
\tikzstyle{wide copoint}=[fill=white,draw=black,shape=isosceles triangle,shape border rotate=-90,isosceles triangle stretches=true,inner sep=1pt,minimum width=1.5cm,minimum height=5mm]
\tikzstyle{symm}=[ultra thick,shorten <=-1mm,shorten >=-1mm]

\tikzstyle{square box}=[rectangle,fill=white,draw=black,minimum height=6mm,minimum width=6mm]
\tikzstyle{square gray box}=[rectangle,fill=gray!30,draw=black,minimum height=6mm,minimum width=6mm]
\tikzstyle{point}=[regular polygon,regular polygon sides=3,draw=black,scale=0.75,inner sep=-0.5pt,minimum width=7mm,fill=white]
\tikzstyle{copoint}=[point,regular polygon rotate=180,fill=white]
\tikzstyle{gray point}=[point,fill=gray!40!white]
\tikzstyle{gray copoint}=[copoint,fill=gray!40!white]

\tikzstyle{open graph}=[baseline=-0.25em]
\tikzstyle{greybg}=[background rectangle/.style={fill=black!5,draw=black!30,rounded corners=1ex}, show background rectangle]

\tikzstyle{edge point}=[circle,minimum width=1mm,fill=black,inner sep=0mm]
\tikzstyle{vertex point}=[circle,minimum width=2.2mm,fill=white,draw=black,inner sep=0mm]
\tikzstyle{gray vertex point}=[circle,minimum width=2.2mm,fill=gray!30!white,draw=black,inner sep=0mm]
\tikzstyle{edge label}=[inner sep=2pt, font=\small]
\tikzstyle{on edge label}=[fill=white, font=\footnotesize, inner sep=1 pt]

\newcommand{\edgearrow}{{\arrow[black]{>}}}
\newcommand{\edgetick}{{\arrow[black,scale=0.7,very thick]{|}}}
\tikzstyle{directed}=[postaction=decorate,decoration={markings, mark=at position 0.55 with \edgearrow}]
\tikzstyle{medium directed}=[postaction=decorate,decoration={markings, mark=at position 0.75 with \edgearrow}]
\tikzstyle{short directed}=[->]
\tikzstyle{halfedge}=[-)]
\tikzstyle{other halfedge}=[(-]
\tikzstyle{freeedge}=[(-)]
\tikzstyle{white edge}=[line width=5pt,white]
\tikzstyle{tick}=[postaction=decorate,decoration={markings, mark=at position 0.5 with \edgetick}]
\tikzstyle{small map edge}=[|-latex, gray!60!blue, shorten <=0.9mm, shorten >=0.5mm]
\tikzstyle{thick dashed edge}=[very thick,dashed,gray!40]
\tikzstyle{dashed edge}=[densely dotted,thick]
\tikzstyle{map edge}=[|-latex,very thick, gray!40, shorten <=1mm, shorten >=0.5mm]
\tikzstyle{tickedge}=[postaction=decorate,
  decoration={markings, mark=at position 0.5 with \edgetick}]
\tikzstyle{dirtickedge}=[postaction=decorate,
  decoration={markings, mark=at position 0.5 with \edgetick},
  decoration={markings, mark=at position 0.85 with \edgearrow}]
\tikzstyle{dirdoubletickedge}=[postaction=decorate,
  decoration={markings, mark=at position 0.4 with \edgetick},
  decoration={markings, mark=at position 0.6 with \edgetick},
  decoration={markings, mark=at position 0.85 with \edgearrow}]

\tikzstyle{arrs}=[-latex,font=\small,auto]
\tikzstyle{arrow plain}=[arrs]
\tikzstyle{arrow dashed}=[dashed,arrs]
\tikzstyle{arrow bold}=[very thick,arrs]
\tikzstyle{arrow hide}=[draw=white!0,-]
\tikzstyle{arrow reverse}=[latex-]
\tikzstyle{cdnode}=[]

\tikzstyle{cnot}=[fill=white,shape=circle,inner sep=-1.4pt]
\tikzstyle{bang box}=[draw=black,dashed,minimum height=12mm,minimum width=12mm,fill=gray!20]
\tikzstyle{wire label}=[font=\footnotesize, auto]

%% file: defs.tex
\tikzstyle{cdiag}=[matrix of math nodes, row sep=3em, column sep=3em, text height=1.5ex, text depth=0.25ex,inner sep=0.5em]
\tikzstyle{arrow above}=[transform canvas={yshift=0.5ex}]
\tikzstyle{arrow below}=[transform canvas={yshift=-0.5ex}]



%% file: figures/Induc-Multiply.tikz
\begin{tikzpicture}[small]
	\begin{pgfonlayer}{nodelayer}
		\node [style=none] (0) at (0, 0.5) {};
		\node [style=none] (1) at (0.25, -0.5) {};
		\node [style=nodevert] (2) at (0, 0) {};
		\node [style=none] (3) at (-0.25, -0.5) {};
	\end{pgfonlayer}
	\begin{pgfonlayer}{edgelayer}
		\draw [style=directed] (1) to (2);
		\draw [style=directed] (3) to (2);
		\draw [style=directed] (2) to (0);
	\end{pgfonlayer}
\end{tikzpicture}

%% file: figures/Induc-Unit.tikz
\begin{tikzpicture}[small]
	\begin{pgfonlayer}{nodelayer}
		\node [style=nodevert] (0) at (0, -0.25) {};
		\node [style=none] (1) at (0, 0.5) {};
	\end{pgfonlayer}
	\begin{pgfonlayer}{edgelayer}
		\draw [style=directed] (0) to (1);
	\end{pgfonlayer}
\end{tikzpicture}

%% file: figures/box_eqn.tikz
\begin{tikzpicture}
	\begin{pgfonlayer}{nodelayer}
		\node [style=nodevert] (0) at (-2.25, -0.5) {};
		\node [style=wirevert] (1) at (-0.75, -1.25) {};
		\node [style=wirevert] (2) at (-2.25, -1.25) {};
		\node [style=written] (3) at (0, 0) {$=$};
		\node [style=nodevert] (4) at (-1.5, 0.75) {};
		\node [style=wirevert] (5) at (-2, 0.25) {};
		\node [style=wirevert] (6) at (1.75, 1.5) {};
		\node [style=nodevert] (7) at (1.75, 0.5) {};
		\node [style=wirevert] (8) at (1, -1.25) {};
		\node [style=wirevert] (9) at (-1.5, 1.5) {};
		\node [style=wirevert] (10) at (2.5, -1.25) {};
		\node [style=none] (11) at (-2, -1.5) {};
		\node [style=none] (12) at (-2, -1) {};
		\node [style=bbox, label={$A$}] (13) at (-2.5, -1) {};
		\node [style=none] (14) at (-2.5, -1.5) {};
		\node [style=none] (15) at (-0.5, -1.5) {};
		\node [style=none] (16) at (-0.5, -1) {};
		\node [style=bbox, label={$B$}] (17) at (-1, -1) {};
		\node [style=none] (18) at (-1, -1.5) {};
		\node [style=none] (19) at (1.25, -1.5) {};
		\node [style=none] (20) at (1.25, -1) {};
		\node [style=bbox, label={$A$}] (21) at (0.75, -1) {};
		\node [style=none] (22) at (0.75, -1.5) {};
		\node [style=none] (23) at (2.75, -1.5) {};
		\node [style=none] (24) at (2.75, -1) {};
		\node [style=bbox, label={$B$}] (25) at (2.25, -1) {};
		\node [style=none] (26) at (2.25, -1.5) {};
	\end{pgfonlayer}
	\begin{pgfonlayer}{edgelayer}
		\draw [style=string] (4) to (9);
		\draw [style=string, in=-30, out=90, looseness=1.00] (1) to (4);
		\draw [style=string] (2) to (0);
		\draw [style=string, bend left=15, looseness=1.00] (0) to (5);
		\draw [style=string, bend left=15, looseness=1.00] (5) to (4);
		\draw [style=string] (7) to (6);
		\draw [style=string, in=-45, out=90, looseness=1.00] (10) to (7);
		\draw [style=string, in=-135, out=90, looseness=1.00] (8) to (7);
		\draw [style=boxedge] (12.center) to (13);
		\draw [style=boxedge] (13) to (14.center);
		\draw [style=boxedge] (14.center) to (11.center);
		\draw [style=boxedge] (11.center) to (12.center);
		\draw [style=boxedge] (16.center) to (17);
		\draw [style=boxedge] (17) to (18.center);
		\draw [style=boxedge] (18.center) to (15.center);
		\draw [style=boxedge] (15.center) to (16.center);
		\draw [style=boxedge] (20.center) to (21);
		\draw [style=boxedge] (21) to (22.center);
		\draw [style=boxedge] (22.center) to (19.center);
		\draw [style=boxedge] (19.center) to (20.center);
		\draw [style=boxedge] (24.center) to (25);
		\draw [style=boxedge] (25) to (26.center);
		\draw [style=boxedge] (26.center) to (23.center);
		\draw [style=boxedge] (23.center) to (24.center);
	\end{pgfonlayer}
\end{tikzpicture}

%% file: figures/box_eqn_inst.tikz
\begin{tikzpicture}
	\begin{pgfonlayer}{nodelayer}
		\node [style=nodevert] (0) at (-4.75, -0.75) {};
		\node [style=wirevert] (1) at (-3.75, -1.25) {};
		\node [style=written] (2) at (-3.25, 0) {$=$};
		\node [style=nodevert] (3) at (-4.25, 0.75) {};
		\node [style=wirevert] (4) at (-4.75, 0) {};
		\node [style=wirevert] (5) at (-2.5, 1.5) {};
		\node [style=nodevert] (6) at (-2.5, 0.5) {};
		\node [style=wirevert] (7) at (-4.25, 1.5) {};
		\node [style=written] (8) at (-1.75, -0.75) {$,$};
		\node [style=wirevert] (9) at (-2.5, -1.25) {};
		\node [style=written] (10) at (-7, 0) {$=$};
		\node [style=nodevert] (11) at (-7.75, 0.75) {};
		\node [style=wirevert] (12) at (-6.25, 1.5) {};
		\node [style=wirevert] (13) at (-7.75, 0) {};
		\node [style=wirevert] (14) at (-7.75, 1.5) {};
		\node [style=written] (15) at (-6, -0.75) {$,$};
		\node [style=nodevert] (16) at (-6.25, 0.5) {};
		\node [style=nodevert] (17) at (-7.75, -0.75) {};
		\node [style=wirevert] (18) at (1, -1.25) {};
		\node [style=written] (19) at (0.25, 0) {$=$};
		\node [style=wirevert] (20) at (-0.5, -1.25) {};
		\node [style=nodevert] (21) at (-0.5, 0.75) {};
		\node [style=wirevert] (22) at (1, 1.5) {};
		\node [style=wirevert] (23) at (-0.5, 0.25) {};
		\node [style=wirevert] (24) at (-0.5, 1.5) {};
		\node [style=written] (25) at (1.75, -0.75) {$,$};
		\node [style=nodevert] (26) at (1, 0.5) {};
		\node [style=nodevert] (27) at (-0.5, -0.5) {};
		\node [style=wirevert] (28) at (6.25, -1.25) {};
		\node [style=wirevert] (29) at (5.25, -1.25) {};
		\node [style=written] (30) at (4.75, 0) {$=$};
		\node [style=wirevert] (31) at (4.25, -1.25) {};
		\node [style=wirevert] (32) at (3, -1.25) {};
		\node [style=nodevert] (33) at (3.75, 0.75) {};
		\node [style=wirevert] (34) at (5.75, 1.5) {};
		\node [style=wirevert] (35) at (3.25, 0.25) {};
		\node [style=wirevert] (36) at (3.75, 1.5) {};
		\node [style=nodevert] (37) at (5.75, 0.5) {};
		\node [style=nodevert] (38) at (3, -0.5) {};
		\node [style=written] (39) at (7.25, -0.75) {$,$};
		\node [style=written] (40) at (8, -0.75) {$\ldots$};
	\end{pgfonlayer}
	\begin{pgfonlayer}{edgelayer}
		\draw [style=string] (3) to (7);
		\draw [style=string, in=-45, out=90, looseness=1.00] (1) to (3);
		\draw [style=string] (0) to (4);
		\draw [style=string, bend left, looseness=1.00] (4) to (3);
		\draw [style=string] (6) to (5);
		\draw [style=string] (9) to (6);
		\draw [style=string] (11) to (14);
		\draw [style=string] (17) to (13);
		\draw [style=string] (13) to (11);
		\draw [style=string] (16) to (12);
		\draw [style=string] (21) to (24);
		\draw [style=string] (20) to (27);
		\draw [style=string] (27) to (23);
		\draw [style=string] (23) to (21);
		\draw [style=string] (26) to (22);
		\draw [style=string] (18) to (26);
		\draw [style=string] (33) to (36);
		\draw [style=string, in=-60, out=90, looseness=1.00] (31) to (33);
		\draw [style=string] (32) to (38);
		\draw [style=string, bend left=15, looseness=1.00] (38) to (35);
		\draw [style=string, bend left=15, looseness=1.00] (35) to (33);
		\draw [style=string] (37) to (34);
		\draw [style=string, in=-60, out=90, looseness=1.00] (28) to (37);
		\draw [style=string, in=-120, out=90, looseness=1.00] (29) to (37);
	\end{pgfonlayer}
\end{tikzpicture}

%% file: figures/tensor_eg1.tikz
\begin{tikzpicture}
	\begin{pgfonlayer}{nodelayer}
		\node [style=arbi] (0) at (0, -0.75) {$\phi$};
		\node [style=wire, label={right:$c$}] (1) at (0.75, -1.5) {};
		\node [style=wire, label={[yshift=5pt]right:$d$}] (2) at (0, -2) {};
		\node [style=wire, label={left:$e$}] (3) at (-0.75, -1.5) {};
		\node [style=arbi] (4) at (0, 1) {$\psi$};
		\node [style=wire, label={[yshift=-5pt]right:$f$}] (5) at (0, 2) {};
	\end{pgfonlayer}
	\begin{pgfonlayer}{edgelayer}
		\draw [style=directed] (1) to (0);
		\draw [style=directed] (2) to (0);
		\draw [style=directed] (3) to (0);
		\draw [style=directed, label=$a$, red, in=-135, out=45, looseness=1.00] (0) to node[left, pos=0.7]{$b$} (4);
		\draw [style=directed, label=$b$, red, in=-45, out=135, looseness=1.00] (0) to node[right, pos=0.7]{$a$} (4);
		\draw [style=directed] (4) to (5);
	\end{pgfonlayer}
\end{tikzpicture}

%% file: figures/tensor_eg2.tikz
\begin{tikzpicture}
	\begin{pgfonlayer}{nodelayer}
		\node [style=arbi] (0) at (-0.5, 0.5) {$\phi$};
		\node [style=arbi] (1) at (1, 0.25) {$\psi$};
		\node [style=arbi] (2) at (-0.5, 2.25) {$\xi$};
		\node [style=arbi] (3) at (0, -1.5) {$\zeta$};
		\node [style=wire, label={right:$e$}] (4) at (0, -2.5) {};
		\node [style=bbox, label={$B$}] (5) at (-1.5, 1.25) {};
		\node [style=none] (6) at (1.75, -0.5) {};
		\node [style=none] (7) at (1.75, 1.25) {};
		\node [style=none] (8) at (-1.5, -0.5) {};
	\end{pgfonlayer}
	\begin{pgfonlayer}{edgelayer}
		\draw [style=directed, arcin={{}{0}{60}{1.5mm}{1}}] (0) to (2);
		\draw [style=directed, arcin={{}{25}{-110}{1.5mm}{1}}] (0) to (3);
		\draw [style=directed] (3) to (1);
		\draw [style=directed] (1) to (0);
		\draw [style=directed] (4) to (3);
		\draw [style=boxedge] (5) to (7.center);
		\draw [style=boxedge] (7.center) to (6.center);
		\draw [style=boxedge] (6.center) to (8.center);
		\draw [style=boxedge] (8.center) to (5);
	\end{pgfonlayer}
\end{tikzpicture}

%% file: figures/btypegraph_eg.tikz
\begin{tikzpicture}
	\begin{pgfonlayer}{nodelayer}
		\node [style=nodevert] (0) at (-2.5, 0) {$f$};
		\node [style=wirevert] (1) at (-0.5, 0) {};
		\node [style=nodevert] (2) at (1.5, 0) {$g$};
		\node [style=wirevert] (3) at (3.5, 0) {};
		\node [style=boxvert] (4) at (3, 1.25) {};
		\node [style=written] (5) at (-0.75, 0.5) {\tiny$X$};
		\node [style=written] (6) at (4, 0) {\tiny$Y$};
	\end{pgfonlayer}
	\begin{pgfonlayer}{edgelayer}
		\draw [style=string, bend right, looseness=1.00] (3) to node[above]{\tiny$\bullet$} (2);
		\draw [style=string, bend left, looseness=1.00] (2) to node[below]{\tiny$\bullet$} (1);
		\draw [style=string, bend left, looseness=1.00] (1) to node[below]{\tiny$\infty$} (0);
		\draw [style=string, bend left, looseness=1.00] (0) to node[above]{\tiny$\bullet$} (1);
		\draw [style=string, bend left, looseness=1.00] (1) to node[above]{\tiny$\bullet$} (2);
		\draw [style=boxstring, in=75, out=148, looseness=0.75] (4) to (0);
		\draw [style=boxstring, in=75, out=165, looseness=0.75] (4) to (1);
		\draw [style=boxstring, in=75, out=-165, looseness=1.00] (4) to (2);
		\draw [style=boxstring, in=90, out=-45, looseness=1.00] (4) to (3);
		\draw [style=boxstring, in=90, out=0, loop] (4) to ();
		\draw [style=string, in=-120, out=-60, loop] (1) to ();
		\draw [style=string, in=-120, out=-60, loop] (3) to ();
	\end{pgfonlayer}
\end{tikzpicture}

%% file: figures/string_eg_1.tikz
\begin{tikzpicture}
	\begin{pgfonlayer}{nodelayer}
		\node [style=nodevert] (0) at (0, 0.5) {};
		\node [style=wirevert, label={right:$a$}] (1) at (0, -0.5) {};
		\node [style=boxvert, label={$A$}] (2) at (-0.75, 0.25) {};
	\end{pgfonlayer}
	\begin{pgfonlayer}{edgelayer}
		\draw [style=string] (1) to (0);
		\draw [style=boxstring] (2) to (1);
	\end{pgfonlayer}
\end{tikzpicture}

%% file: figures/string_eg_2.tikz
\begin{tikzpicture}
	\begin{pgfonlayer}{nodelayer}
		\node [style=nodevert] (0) at (0, 0) {};
		\node [style=wirevert, label={right:$a$}] (1) at (0, 0.75) {};
		\node [style=wirevert, label={right:$b$}] (2) at (0, -0.75) {};
		\node [style=boxvertal, label={$B$}] (3) at (-0.75, -0.5) {};
	\end{pgfonlayer}
	\begin{pgfonlayer}{edgelayer}
		\draw [style=string] (2) to (0);
		\draw [style=string] (0) to (1);
		\draw [style=boxstring] (3) to (2);
	\end{pgfonlayer}
\end{tikzpicture}

%% file: figures/string_eg_3.tikz
\begin{tikzpicture}
\end{tikzpicture}

%% file: figures/string_eg_4.tikz
\begin{tikzpicture}
	\begin{pgfonlayer}{nodelayer}
		\node [style=wirevert, label={right:$c$}] (0) at (0, -0.75) {};
		\node [style=wirevert, label={right:$d$}] (1) at (0, 0.75) {};
	\end{pgfonlayer}
	\begin{pgfonlayer}{edgelayer}
		\draw [style=string] (0) to (1);
	\end{pgfonlayer}
\end{tikzpicture}

%% file: figures/ncq_edge_order.tikz
\begin{tikzpicture}
	\begin{pgfonlayer}{nodelayer}
		\node [style=arbi, scale=3, line width=1pt] (0) at (0, 0) {};
		\node [style=wire, scale=2] (1) at (2.75, -3) {};
		\node [style=wire, scale=2] (2) at (-0.75, 3.5) {};
		\node [style=wire, scale=2] (3) at (0.75, 3.5) {};
		\node [style=handle] (4) at (0, -2.5) {};
		\node [style=handle] (5) at (1.5, 2) {};
		\node [style=handle] (6) at (-1.5, 2) {};
		\node [style=handle] (7) at (-1.15, 0) {};
		\node [style=written] (8) at (0.5, 2.5) {};
		\node [style=written] (9) at (2.25, 1.25) {};
	\end{pgfonlayer}
	\begin{pgfonlayer}{edgelayer}
		\draw [style=directed, line width=1pt, in=-90, out=150, looseness=1.00] (1) to (0);
		\draw [style=handleedge] (0) to (6);
		\draw [style=handleedge] (0) to (5);
		\draw [style=handleedge] (0) to (4);
		\draw [style=movement, <->, bend left=15, looseness=1.00] (8) to (9);
		\draw [style=directed, line width=1pt, in=-90, out=60, looseness=1.00] (0) to (2);
		\draw [style=directed, line width=1pt, in=-90, out=120, looseness=1.00] (0) to (3);
	\end{pgfonlayer}
\end{tikzpicture}

%% file: figures/ncq_arc_drawing_2.tikz
\begin{tikzpicture}
	\begin{pgfonlayer}{nodelayer}
		\node [style=arbi, scale=3, line width=1pt] (0) at (-3, 0) {};
		\node [style=wire, scale=2] (1) at (2, -1) {};
		\node [style=wire, scale=2] (2) at (2, 0.5) {};
		\node [style=wire, scale=2] (3) at (1.25, 2.5) {};
		\node [style=handle] (4) at (-2.75, 1.45) {};
		\node [style=handle] (5) at (-1.3, -1.35) {};
		\node [style=bbox, scale=2] (6) at (0, 3.25) {};
		\node [style=none] (7) at (0, -2.25) {};
		\node [style=none] (8) at (2.5, -2.25) {};
		\node [style=none] (9) at (2.5, 3.25) {};
		\node [style=none] (10) at (2.5, 1.5) {};
		\node [style=none] (11) at (1, -1.75) {};
		\node [style=none] (12) at (2.5, -1.75) {};
		\node [style=bbox, scale=2] (13) at (1, 1.5) {};
		\node [style=none] (14) at (-2.25, -0.75) {\includegraphics[scale=0.015]{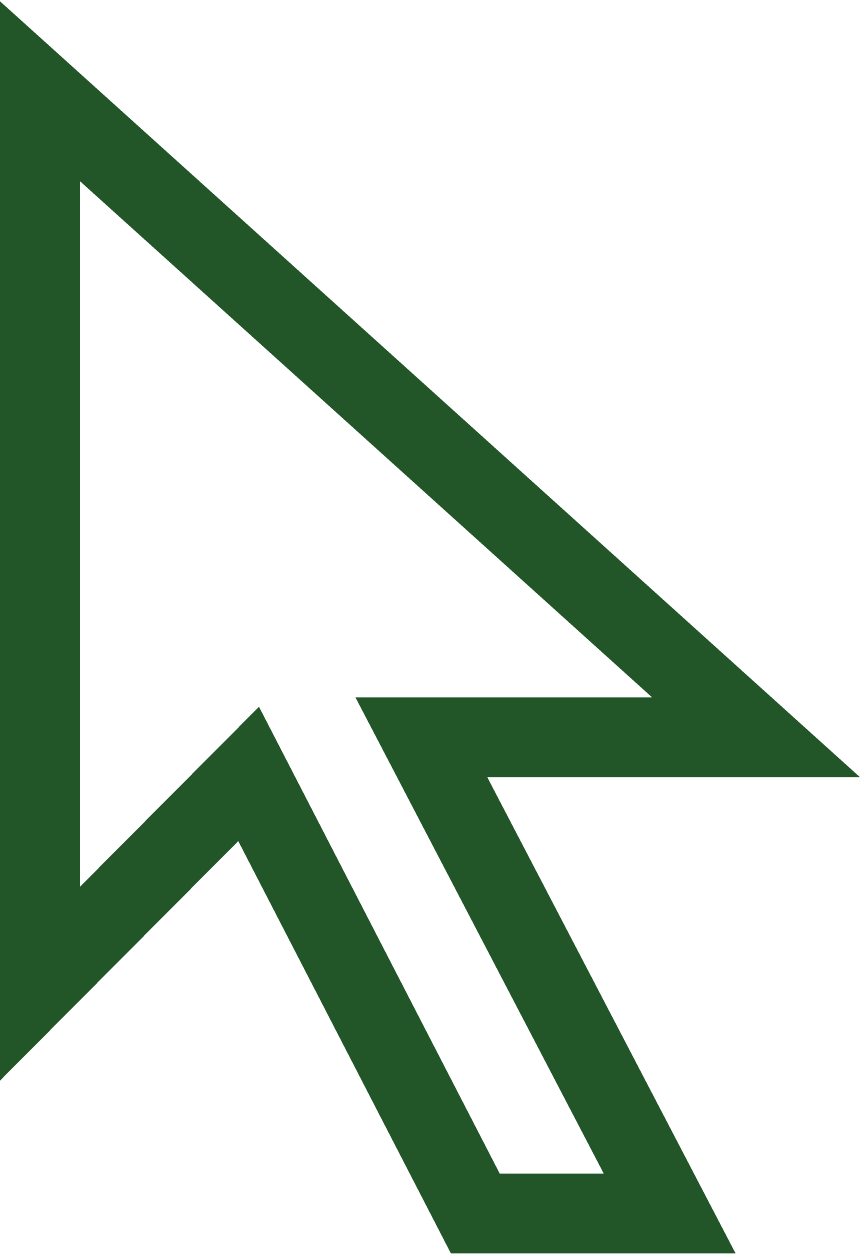}};
		\node [style=none] (15) at (-0.5, -0.5) {\includegraphics[scale=0.015]{mouse}};
	\end{pgfonlayer}
	\begin{pgfonlayer}{edgelayer}
		\draw [style=directed, line width=1pt, arcin={{}{0}{-35}{8mm}{3}}, in=-30, out=-165, looseness=1.00] (1) to (0);
		\draw [style=directed, line width=1pt, arcout={{}{10}{35}{8mm}{3}}, in=150, out=30, looseness=0.75] (0) to (2);
		\draw [style=directed, line width=1pt, arcout={{}{60}{160}{4mm}{3}}, in=180, out=75, looseness=0.75] (0) to (3);
		\draw [style=boxedge, line width=0.75pt] (9.center) to (6);
		\draw [style=boxedge, line width=0.75pt] (6) to (7.center);
		\draw [style=boxedge, line width=0.75pt] (7.center) to (8.center);
		\draw [style=boxedge, line width=0.75pt] (12.center) to (11.center);
		\draw [style=boxedge, line width=0.75pt] (11.center) to (13);
		\draw [style=boxedge, line width=0.75pt] (13) to (10.center);
		\draw [style=movement, in=60, out=-30, looseness=0.75] (4) to (14.center);
		\draw [style=movement, in=-90, out=-45, looseness=1.50] (5) to (15.center);
	\end{pgfonlayer}
\end{tikzpicture}